\documentclass[reqno,a4paper]{llncs}
\pdfoutput=1
\usepackage[utf8]{inputenc}
\usepackage[english]{babel}

\usepackage{setspace}
\usepackage{amsmath}
\usepackage{amssymb}
\usepackage{graphicx}
\usepackage{datetime}
\usepackage{enumitem}

\def\rmlabel{\upshape({\itshape \roman*\,})}

\usepackage{mathabx}
\usepackage{mathrsfs} 
\usepackage{dsfont} 
\usepackage[babel]{microtype}

\usepackage{xcolor} 
\usepackage{color} 	
\usepackage[backref]{hyperref}

\usepackage[T1]{fontenc}

\let\emptyset=\varnothing
\let\setminus=\smallsetminus

\makeatletter
\def\moverlay{\mathpalette\mov@rlay}
\def\mov@rlay#1#2{\leavevmode\vtop{   \baselineskip\z@skip \lineskiplimit-\maxdimen
   \ialign{\hfil$\m@th#1##$\hfil\cr#2\crcr}}}
\newcommand{\charfusion}[3][\mathord]{
    #1{\ifx#1\mathop\vphantom{#2}\fi
        \mathpalette\mov@rlay{#2\cr#3}
      }
    \ifx#1\mathop\expandafter\displaylimits\fi}
\makeatother

\newcommand{\quot}[2]{\mathchoice%
{\left.\raisebox{.1em}{$\displaystyle{#1}$}\kern-1pt/\raisebox{-.2em}{$\displaystyle{#2}$}\right.}
{\left.\raisebox{.1em}{${#1}$}\kern-1pt/\raisebox{-.2em}{${#2}$}\right.}
{\left.\raisebox{.1em}{$\scriptstyle{#1}$}\kern-1pt/\raisebox{-.2em}{$\scriptstyle{#2}$}\right.}
{\left.\raisebox{.1em}{$\scriptscriptstyle{#1}$}\kern-1pt/\raisebox{-.2em}{$\scriptscriptstyle{#2}$}\right.}
}

\DeclareFontFamily{U}  {MnSymbolC}{}
\DeclareSymbolFont{MnSyC}         {U}  {MnSymbolC}{m}{n}
\DeclareFontShape{U}{MnSymbolC}{m}{n}{
    <-6>  MnSymbolC5
   <6-7>  MnSymbolC6
   <7-8>  MnSymbolC7
   <8-9>  MnSymbolC8
   <9-10> MnSymbolC9
  <10-12> MnSymbolC10
  <12->   MnSymbolC12}{}
\DeclareMathSymbol{\powerset}{\mathord}{MnSyC}{180}

\newtheorem{fact}       [theorem] {Fact}

\let\epsilon\varepsilon
\let\hat\widehat
\newcommand{\EE}{\mathbb{E}}
\newcommand{\PP}{\mathbb{P}}
\newcommand{\ZZ}{\mathbb{Z}}
\newcommand{\RR}{\mathbb{R}}
\newcommand{\DD}{\mathbb{D}}
\newcommand{\cA}{\mathcal{A}}

\newcommand{\cG}{\mathcal{G}}
\newcommand{\cH}{\mathcal{H}}
\newcommand{\cI}{\mathcal{I}}
\newcommand{\cL}{\mathcal{L}}
\newcommand{\cN}{\mathcal{N}}
\newcommand{\cSN}{\mathcal{SN}}
\newcommand{\cP}{\mathcal{P}}
\newcommand{\cR}{\mathcal{R}}

\newcommand{\cU}{\mathcal{U}}

\newcommand{\opt}{{\rm OPT}}
\newcommand{\soc}{{\rm SC}}

\def\({\left(}
\def\){\right)}  
\def\[{\left[}
\def\]{\right]}

\def\llceil{\left\lceil}
\def\rrceil{\right\rceil}  
\def\:{\colon}
\def\ee{{\rm e}}

\usepackage{lineno}
\newcommand*\patchAmsMathEnvironmentForLineno[1]{%
\expandafter\let\csname old#1\expandafter\endcsname\csname #1\endcsname
\expandafter\let\csname oldend#1\expandafter\endcsname\csname end#1\endcsname
\renewenvironment{#1}%
{\linenomath\csname old#1\endcsname}%
{\csname oldend#1\endcsname\endlinenomath}}%
\newcommand*\patchBothAmsMathEnvironmentsForLineno[1]{%
\patchAmsMathEnvironmentForLineno{#1}%
\patchAmsMathEnvironmentForLineno{#1*}}%
\AtBeginDocument{%
\patchBothAmsMathEnvironmentsForLineno{equation}%
\patchBothAmsMathEnvironmentsForLineno{align}%
\patchBothAmsMathEnvironmentsForLineno{flalign}%
\patchBothAmsMathEnvironmentsForLineno{alignat}%
\patchBothAmsMathEnvironmentsForLineno{gather}%
\patchBothAmsMathEnvironmentsForLineno{multline}%
}

\allowdisplaybreaks[2]

\def\PoA{\mathop{{\rm PoA}}\nolimits}
\def\SPoA{\mathop{{\rm SPoA}}\nolimits}
\def\val{\mathop{{\rm w}}\nolimits}

\begin{document}

\title{A tight lower bound for an online hypercube \\ 
       packing problem and bounds for prices of \\ 
       anarchy of a related game}

\titlerunning{Bounds for an online hypercube packing and a related
game}  

\author{Y. Kohayakawa\inst{1} \and
  F.K. Miyazawa\inst{2} \and 
   Y. Wakabayashi\inst{1}}

\authorrunning{Y.\ Kohayakawa, F.K. Miyazawa and Y.\ Wakabayashi}

\institute{Institute of Mathematics and Statistics \\
University of S\~ao Paulo, Brazil \\
\email{\{yoshi|yw\}@ime.usp.br}\\
\and
Institute of Computing \\
University of Campinas, Brazil\\
\email{fkm@ic.unicamp.br}
}

\maketitle

\begin{abstract}
  
We prove a tight lower bound on the asymptotic performance
ratio~$\rho$ of the \emph{bounded space online $d$-hypercube bin packing
problem}, solving an open question raised in 2005.
In the classic $d$-hypercube bin packing problem, we are given a sequence
of $d$-dimensional hypercubes 
and we have an unlimited number of bins, each of which is a
$d$-dimensional unit hypercube. The goal is to pack (orthogonally) the
given hypercubes into the minimum possible number of bins, in such a
way that no two hypercubes in the same bin overlap. The \emph{bounded
  space online \hbox{$d$-hypercube} bin packing problem} is a variant
of the $d$-hypercube bin packing problem, in which the hypercubes arrive
\emph{online} and each one must be packed in an open bin without the knowledge
of the next hypercubes.  Moreover, at each moment, only a constant
number of open bins are allowed (whenever a new bin is used, it is
considered open, and it remains so until it is considered closed, in which case,
it is not allowed to accept new hypercubes).  Epstein and van
Stee~[SIAM J.\ Comput.~35 (2005), no.~2, 431--448] showed that~$\rho$
is $\Omega(\log d)$ and $O(d/\log d)$, and conjectured that it is
$\Theta(\log d)$. We show that $\rho$ is in fact $\Theta(d/\log d)$.
  To obtain this result, we elaborate on some ideas presented by those
  authors, and go one step further showing how to obtain better
  (offline) packings of certain special instances for which one knows
  how many bins any bounded space algorithm has to use. Our main
  contribution establishes the existence of such packings,
  for large enough~$d$, using probabilistic arguments.  Such packings
  also lead to lower bounds for the prices of anarchy of the selfish
  $d$-hypercube bin packing game.  We present a lower bound of
  $\Omega(d/\log d)$ for the pure price of anarchy of this game, and
  we also give a lower bound of~$\Omega(\log d)$ for its strong price
  of anarchy.

\end{abstract}

\shortdate
\yyyymmdddate
\def\today{\number\year/\number\month/\number\day}
\settimeformat{ampmtime}
\date{\today, \currenttime}


\section{Introduction and main results}
\label{sec:introduction}

The bin packing problem is an iconic  problem in combinatorial
optimization that has been largely investigated from many
different viewpoints. In special, it has served as a proving ground for new
approaches to the analysis of approximation algorithms. It is one of
the first problems for which approximation algorithms were proposed in
the beginning of seventies, and also ideas to prove lower bounds for
online algorithms and probabilistic analysis first
appeared~\cite{CoffmanGJ97}. We believe that the technique we present
in this paper is novel and contributes with new ideas that may possibly be
incorporated into this area of research.

We prove bounds for two variants of the bin packing problem, in which
the items to be packed are $d$-dimensional cubes (also referred to as
\emph{$d$-hypercubes} or simply hypercubes, when the dimension is
clear).  More precisely, we show results for the \emph{online bounded
  space $d$-hypercube bin packing problem} and the
\emph{selfish hypercube bin packing game}.  Before we state
our results in the next section, we define these problems and mention some
known results.

The \emph{$d$-hypercube bin packing problem} ($d$-CPP) is defined as
follows.  We are given a list $L$ of items, where each item
$h\in L$ is a $d$-hypercube of side length $s(h)\leq1$, and an unlimited
number of bins, each of which is a unit $d$-hypercube. 
The goal is to find a packing $\cP$ of the items of~$L$ into a minimum
number of bins. 
More precisely, we have to assign each item $h$ to a bin, and specify
its position $(x_1(h),\ldots,x_d(h))$ in this bin.  As usual, we
consider that each bin is defined by the region $[0,1]^d$, and thus,
we must have $0\leq x_i(h)\leq 1-s(h)$, for
$i=1,\ldots,d$. Additionally, we must place the items parallel to the
axes of the bin and guarantee that items in the same bin do not
overlap.  The \emph{size} of the packing $\cP$ is the number of
\emph{used} bins (those with at least one item assigned to it). 
Throughout this paper, the bins are
always assumed to be unit hypercubes of the same dimension of the
items that have to be packed.

The $d$-CPP is in fact a special case of the \textit{$d$-dimensional
  bin packing problem} ({\rm {\it d}-BPP}), in which one has to pack
$d$-dimensional parallelepipeds into $d$-dimensional unit bins.  For
$d=1$, both problems reduce to the well known \textit{bin packing problem}.

In the \emph{online} variant of $d$-CPP, the hypercubes arrive online
and must be packed in an open bin (without the knowledge of the next
hypercubes). The {\it online bounded space} variant of the $d$-CPP is
a more restricted variant of the online $d$-CPP. Whenever a new empty
bin is used, it is considered an open bin and it remains so until it
is considered closed, after which it is not allowed to accept other
hypercubes. In this variant, during the packing process, only a
constant number of open bins is allowed.  The corresponding problem or
algorithm in which the whole list of items is known beforehand is
called offline.

As it is usual, for bin packing problems, we consider the asymptotic
performance ratio to measure the quality of the algorithms. For an
algorithm $\cA$, and an input list~$L$, let $A(L)$ be the number of bins
used by the solution produced by algorithm $A$ for the list $L$, and
let $\opt(L)$ be the minimum number of bins needed to pack $L$. The
\emph{asymptotic performance ratio} of algorithm $A$ is defined as
\begin{equation}
\cR_A^{\infty} = \limsup_{n\rightarrow\infty}
               \sup_L\left\{  \frac{A(L)}{\opt(L)}: \opt(L)=n\right\}.
\end{equation}
Given a packing problem $\Pi$, the {\it optimal asymptotic performance
  ratio} for $\Pi$ is defined as 
\begin{equation}
  \cR_{\Pi}^{\infty}=\inf\left\{\cR_A^{\infty}: A\mbox{ is an
    algorithm for }\Pi\right\}.
\end{equation}


Many results have been obtained for the online $d$-BPP and
  $d$-CPP problems
  (see~\cite{Vliet92,Seiden02,BaloghBG10,HeydrichS2016a,HeydrichS16b,BaloghBDEL17arx-a}). Owing to space limitation, we mention only results for the online bounded
  space versions of these problems.
For the online bounded space $1$-BPP,
  Csirik~\cite{Csirik89} presented an algorithm with asymptotic
  performance ratio at most $\Pi_{\infty}\approx 1.69103$, shown to be
  an optimal online bounded space algorithm by
  Seiden~\cite{Seiden01}. For the online bounded space $d$-BPP, $d\geq 2$, a lower bound of~$(\Pi_{\infty})^d$ follows from~\cite{CsirikV93}; Epstein and van Stee~\cite{EpsteinS05} showed that this bound is tight.

For the online bounded space $d$-CPP, Epstein and van
Stee~\cite{EpsteinS05} showed that its asymptotic performance ratio is
$\Omega(\log d)$ and $O(d/\log d)$, and conjectured that it is
$\Theta(\log d)$. They also showed an optimal algorithm for this
problem, but left as an interesting open problem to determine its
asymptotic performance ratio.  One of our main results builds
  upon their work and shows a lower bound that matches the known upper
  bound. 
\begin{theorem}
  \label{thm:lwbd_prbsa}
  The asymptotic performance ratio of the online bounded space $d$-hypercube
  bin packing problem is~$\Omega(d/\log d)$.
\end{theorem}
In view of the previous results~\cite{EpsteinS05}, we have that
  the asymptotic performance ratio of the online bounded space
  $d$-hypercube bin packing problem is  $\Theta(d/\log d)$. 
  Results on lower and upper bounds for $d\in\{2,\ldots,7\}$ have also
  been obtained by Epstein and van Stee~\cite{EpsteinS07}.

The technique that we use to prove the above theorem can also be used
to obtain lower bounds for a game theoretic version of the $d$-CPP
problem, called \emph{selfish $d$-hypercube bin packing game}.

This game starts with a set of $d$-hypercubes arbitrarily packed into unit
bins. Each of these hypercubes is (controlled by) a player. For
simplicity, in the game context, we will use the terms hypercube, item
and player in an interchangeable manner. For a game with $n$ items, a
\emph{configuration} is a vector $p = (p_1, \ldots, p_n)$, where $p_i$
indicates in which bin item $i$ is packed. (Equivalently, a
configuration is a packing of the items into bins.)
The \emph{cost} of an item is defined as the ratio between its volume
and the total occupied volume of the respective bin. In this game,
an item can migrate to another bin only when its cost decreases. 
Players may act selfishly by changing their strategy (that is, moving to another bin)
to minimize their costs. 
For a given game configuration~$p$, its \emph{social cost}, denoted
by~$\soc(p)$, is the total cost paid by the players (which is
precisely the number of used bins).  The \emph{optimal social
  goal} is a game configuration of minimum social cost, which we
denote by $\opt(L)$.

An important concept in game theory is the Nash equilibrium~\cite{Nash51}.
In the selfish hypercube bin packing game, a \emph{(pure) Nash
  equilibrium} is a stable packing where no player can reduce his cost
by unilaterally changing his strategy (that is, moving to another
bin), while the strategies of all other players remain unchanged.
The pure Nash equilibrium may not be resilient to the action of
coalitions, as it does not assume that players negotiate and cooperate
with each other. Aumann~\cite{Aumann59} introduced the concept of
\emph{strong Nash equilibrium} in coalitional game theory; in this
case, a group of players may agree to coordinate their actions in a
mutually beneficial way. A \emph{strong Nash equilibrium} is a game
configuration where no group of players can reduce the cost of each of
its members by changing strategies together, while non-members
maintain their strategies.

Throughout the paper, the Nash equilibrium is considered
only in the setting of pure strategies (for pure strategies, a player
chooses only one strategy at a time, while for mixed strategies, a
player chooses an assignment of probabilities to each pure strategy).
Given a game $G$, we denote by $\cN(G)$ (resp. $\cSN(G)$) the set of
configurations in Nash equilibrium (resp. strong Nash equilibrium).

To measure the quality of an equilibrium, Koutsoupias and
Papadimitriou~\cite{KoutsoupiasP99} proposed a measure
in a game-theoretic framework that nowadays is known as the
\textit{price of anarchy} (resp. \textit{strong price of anarchy}),
which is the ratio between the worst social cost of a Nash equilibrium
(resp. strong Nash equilibrium) and the optimal social cost. The price
of anarchy measures the loss of the overall performance due to the
decentralized environment and the selfish behavior of the players.
As it is common for bin packing problems, for bin packing games one
also considers asymptotic price of anarchy. 
The \emph{(asymptotic) price of anarchy} of a class $\cG$ of games is
defined as
\begin{equation}
  \PoA (\cG) := \limsup_{m\rightarrow\infty\;\;} \sup_{G\in\cG,\; \opt(G)=m\; } \max_{p \in \cN(G)} \frac{\soc(p)}{m}. 
\end{equation}
The \emph{(asymptotic) strong price of anarchy} of a class $\cG$ of
games, denoted $\SPoA (\cG)$, is defined analogously, considering only configurations that are
strong Nash equilibria. 

We are interested in the case $\cG$ is the class of the \emph{selfish
 $d$-hypercube bin packing games}, with the natural cost function
(proportional model) we have defined. (Note that, other cost functions
can also be defined for bin packing games.)  We will prove bounds
for the asymptotic prices of anarchy of this class of games. The
corresponding measures will be denoted by $\PoA(d)$ and $\SPoA(d)$,
where $d$ indicates the dimension of the items in the game. Although
we may not mention explicitly, the prices of anarchy considered are
always asymptotic. 

The case $d=1$ of this game was first investigated by
Bil\`o~\cite{Bilo06}, who referred to it as selfish bin packing
game. He proved that this game always converges to a pure Nash
equilibrium and proved that $\PoA (1) \in [1.6,\,1.666]$.  Yu and
Zhang~\cite{YuZ08} improved this result to
$\PoA(1) \in [1.6416, \, 1.6575]$.  Epstein and
Kleiman~\cite{EpsteinK11} obtained (independently) the same lower
bound and improved the upper bound to $1.6428$; they also proved that
$\SPoA(1) \in [1.6067,\, 1.6210]$. Very recently, Epstein, Kleiman and
Mestre~\cite{EpsteinKM16} showed that $\SPoA(1) \approx 1.6067$.  For
$d=1$, Ma et al.~\cite{MaDHTYZ13}, obtained results considering
another cost function. The case $d=2$ was first investigated by
Fernandes~et~al.~\cite{FernandesFMW12}. They showed
in~\cite{FernandesFMW17} that $\PoA(2) \in [2.3634,\, 2.6875]$ and
$\SPoA(2) \in[2.0747, \, 2.3605]$. For a survey on bin packing games
with selfish items, we refer the reader to Epstein~\cite{Epstein13}.

Our second set of results concern lower and upper bounds  for
$\PoA(d)$ and $\SPoA(d)$. 

\begin{theorem}
  \label{thm:PoA_lwbd}
  Let $\PoA(d)$ be the price of anarchy of the selfish
  $d$-hypercube bin packing game. There is an absolute constant~$d_0$
  such that, for all~$d\geq d_0$, we have
  \begin{equation}
    \label{eq:PoA_lwbd}
    \PoA(d)\geq{d\over5\log d}.
  \end{equation}
\end{theorem}

We remark that our proof of Theorem~\ref{thm:PoA_lwbd}, presented in
Section~\ref{sec:pfofThm.PoA_lwbd}, may be adapted to prove the following statement:
\textsl{for any~$\epsilon>0$ there is~$d_0=d_0(\epsilon)$ such that,
  for any~$d\geq d_0$, we have~$\PoA(d)\geq(1/4-\epsilon)d/\log d$.}

\begin{theorem}
  \label{thm:SPoA_lwbd}
  Let $\SPoA(d)$ be the strong price of anarchy of the selfish
  $d$-hypercube bin packing game. There is an absolute constant~$d_0$
  such that, for all~$d\geq d_0$, we have
  \begin{equation}
    \label{eq:SPoA_lwbd}
    \SPoA(d)\geq{\log d}.
  \end{equation}
 \end{theorem}

 The proof of Theorem~\ref{thm:SPoA_lwbd} uses arguments similar to
 those used in the proof of Theorem~\ref{thm:PoA_lwbd} and is
 therefore omitted (see Appendix~\ref{sec:pfofThm.strong_PoA_lwbd}).
We also prove that the price of anarchy of the selfish
$d$-hypercube bin packing game is at most~$2^d$ (see Appendix~\ref{sec:further}). 
We believe the probabilistic technique used to obtain the lower bounds
in Theorems~\ref{thm:lwbd_prbsa}, \ref{thm:PoA_lwbd}
and~\ref{thm:SPoA_lwbd} is novel and promising for obtaining lower
bounds for other packing problems and games.

\section{Notation, special packings and central lemmas}
\label{sec:preliminaries}
The open $d$-hypercubes~$Q_k^d(\epsilon)$ defined below will be
crucial in what follows.

\begin{definition}
  \label{def:Q_k^d(gepsilon)}
  Let~$d\geq2$ be an integer.  For all integer~$k\geq2$ and~$0<\epsilon\leq1$,
  let
  \begin{equation}
    \label{eq:Q_k^d(epsilon).def}
    Q_k^d(\epsilon)=(1+\epsilon)\(0,{1\over
      k}\)^d=\(0,{1+\epsilon\over k}\)^d
    =\left\{x\in\RR:0<x<{1+\epsilon\over k}\right\}^d\subset[0,1]^d
  \end{equation}
  be the open $d$-hypercube of side length~$(1+\epsilon)/k$
  `based' at the origin.
\end{definition}

For convenience, given~$\epsilon>0$ and a positive integer~$q$, we
write~$q_{-\epsilon}$ for~$q/(1+\epsilon)$.  The quantity~$\epsilon$
will often be clear from the context, and in those cases we simply
write~$q_-$ for~$q_{-\epsilon}$.  Note that, for instance, we
have
\begin{equation}
  \label{eq:Q_k_-^d}
  Q_k^d(\epsilon)=\(0,{1\over k_-}\)^d.
\end{equation}

In what follows, we are interested in certain types of packings~$\cU$
of hypercubes into a unit bin.  If a packing~$\cP$ of hypercubes is
made up of packings~$\cU_1,\dots,\cU_N$, with each~$\cU_i$ being a packing
into a unit bin, then we write~$\cP=(\cU_1,\dots,\cU_N)$, and
denote by~$|\cP|$ the number of bins~$N$ in~$\cP$.

\begin{definition}[Packings of type~$\cH_k^d(\epsilon)$
  and~$\cH^d(\epsilon)$] 
  \label{def:cH}
  Let~$d\geq2$ be fixed.  For any integer $k\geq2$
  and~$0<\epsilon\leq1/(k-1)$, a packing~$\cU$ of $(k-1)^d$ copies
  of~$Q_k^d(\epsilon)$ into a unit bin is said to be a \textit{packing
    of type~$\cH_k^d(\epsilon)$}.  A packing~$\cP=(\cU_1,\cU_2,\dots)$
  is said to be of \textit{type~$\cH^d(\epsilon)$} if for each~$i$
  there is some~$k$ such that~$\cU_i$ is a packing of
  type~$\cH_k^d(\epsilon)$.
\end{definition}

In the definition above, the upper bound on~$\epsilon$ guarantees
that~$(k-1)^d$ copies of~$Q_k^d(\epsilon)$ \textit{can} be packed into
a unit bin (and hence~$\cH_k^d(\epsilon)$ exists): it suffices to
notice that, under that assumption on~$\epsilon$, we have
$(k-1)(1+\epsilon)/k\leq1$.

  Packings of type~$\cH_k^d(\epsilon)$ and~$\cH^d(\epsilon)$ are
  called \textit{homogeneous packings}.  They will be important for us
  because they are Nash equilibria (see
  Lemma~\ref{lem:an_equilibrium}), and also because they can be used
  to create instances for which any bounded space algorithm performs
  badly (following ideas of Epstein and van Stee~\cite{EpsteinS05,EpsteinS07}).

\subsubsection{Two packing lemmas.}
\label{sec:packing}

For the next definition, suppose~$\DD$ is a given set of integers, and
$\epsilon$ is a positive real number.

\begin{definition}[Packings of type~$(1+\epsilon)\DD^{-1}$]
  \label{def:(1+epsilon)/ZZ}
  A packing~$\cU$ of $d$-hypercubes into a unit bin is of
  \textit{type~$(1+\epsilon)\DD^{-1}$} if, for every member~$Q$
  of~$\cU$, there is some integer~$k\in\DD$ such that~$Q$ is a copy
  of~$Q_k^d(\epsilon)$.
\end{definition}

  In what follows, we shall restrict to packings~$\cU$ of
  type~$(1+\epsilon)\DD^{-1}$, where $\DD$ is one of the following
  sets: (a)~$\DD=\ZZ_{\geq2}=\{k\in\ZZ\:k\geq2\}$ or
  (b)~$\DD=\ZZ_{2^+}$, where ~$\ZZ_{2^+}$ denotes the
  set~$\{2^i\:i\geq1\}$.  
  Following~\cite{EpsteinK11,EpsteinKM16}, we consider~$\DD=\ZZ_{2^+}$ to deal with strong
  Nash equilibria (see Appendix~\ref{sec:pfofThm.strong_PoA_lwbd}).
 
Let~$\cU$ be a packing of type~$(1+\epsilon)\DD^{-1}$ for
some~$\DD\subset\ZZ_{\geq2}$ and~$\epsilon>0$.  Let
\begin{equation}
  \label{eq:K(U)_def}
  K(\cU)=\{k\in\DD\:\cU\text{ contains a copy of }Q_k^d(\epsilon)\}
\end{equation}
and
\begin{equation}
  \label{eq:k_max_def}
  k_{\max}(\cU)=\max\{k\:k\in K(\cU)\}. 
\end{equation}
For every~$k\in K(\cU)$, let 
\begin{equation}
  \label{eq:nu_k(U)_def}
  \nu_k(\cU)\mbox{ be the total number of copies of
  }Q_k^d(\epsilon)\mbox{ in }\cU.
\end{equation}
Clearly, we have~$0\leq\nu_k(\cU)\leq(k-1)^d$ for every~$k$ (recall
that we suppose~$\epsilon>0$).  Finally, we define the \textit{weight} of~$\cU$
as
\begin{equation}
  \label{eq:val_def}
  \val(\cU)=\sum_{k\in K(\cU)}(k-1)^{-d}\nu_k(\cU).
\end{equation}
We shall be interested in packings~$\cU$ with large weight.  In
  that direction, we prove the following two technical results that
  are the core of our contribution. The first is essential to derive
  the lower bound for the online bounded space $d$-hypercube bin
  packing problem (Theorem~\ref{thm:lwbd_prbsa}) and a lower bound
  for~$\PoA$~(Theorem~\ref{thm:PoA_lwbd}); the second is essential to
  derive a lower bound for~$\SPoA$~(Theorem~\ref{thm:SPoA_lwbd}).  We
  remark that the technique of using weight functions in the analysis
  of packing algorithms dates back to the seventies
  (see~\cite{EpsteinS05} and the references therein).

\begin{lemma}[Packing lemma~A]
  \label{lem:large_value_pack}
  There is an absolute constant~$d_0$ for which the following holds
  for any~$d\geq d_0$.  Let
  \begin{equation}
    \label{eq:S_def}
    S=\llceil2d\over9\log d\rrceil.
  \end{equation}
  The unit bin admits a packing~$\cU$ of type
  $(1+S^{-2})\ZZ_{\geq2}^{-1}$ with~$k_{\max}(\cU)=S$ and with
  \begin{equation}
    \label{eq:large_value_pack}
    \val(\cU)\geq{d\over5\log d}.
  \end{equation}
\end{lemma}

\begin{lemma}[Packing lemma B]
  \label{lem:large_value_pack2+}
  There is an absolute constant~$d_0$ for which the following holds
  for any~$d\geq d_0$.  Let
  \begin{equation}
    \label{eq:S'_def.1}
    S'=\llceil\log_2d-\log_2\log d-3\rrceil
  \end{equation}
  and~$\epsilon=2^{-2(S'-1)}$.  The unit bin admits a
  packing~$\cU$ of type $(1+\epsilon)\ZZ_{2^+}^{-1}$
  with~$k_{\max}(\cU)=2^{S'-1}$ and with~$\val(\cU)\geq\log d$.  
\end{lemma}

\section{Proofs of Theorems~\ref{thm:lwbd_prbsa} and~\ref{thm:PoA_lwbd}}
\label{sec:pfofThm.PoA_lwbd}
To prove Theorem~\ref{thm:lwbd_prbsa}, one can use
Lemma~\ref{lem:large_value_pack} to produce suitable instances that
are `hard' for algorithms for the online bounded space $d$-hypercube
bin packing problem; see, e.g.,~\cite[Lemma~2.3]{EpsteinS05}.
A detailed proof is given in Appendix~\ref{sec:pfofThm.lwbd_prbsa}.

To prove Theorem~\ref{thm:PoA_lwbd}, we shall use
Lemma~\ref{lem:large_value_pack} and the next two lemmas, the proofs
of which are presented in Appendix~\ref{app:pfofThm.PoA_lwbd}.

\begin{lemma}
  \label{lem:an_equilibrium}
  Let~$d\geq2$ and~$\epsilon>0$ be given.  Any
  packing~$\cP=(\cU_1,\cU_2,\dots)$ of \textit{type~$\cH^d(\epsilon)$}
  is a Nash equilibrium.
\end{lemma}

\begin{lemma}
  \label{lem:PoA_from_good}
  If~$\cU$ is a packing of $d$-hypercubes into a unit bin of type
  $(1+\epsilon)\ZZ_{\geq2}^{-1}$, where
  \begin{equation}
    \label{eq:epsilon_upbdd}
    0<\epsilon\leq{1\over k_{\max}(\cU)-1},
  \end{equation}
  then~$\PoA(d)\geq\val(\cU)$.
\end{lemma}

We now prove Theorem~\ref{thm:PoA_lwbd}.
%
  Let~$d_0$ be as in Lemma~\ref{lem:large_value_pack} and
  suppose~$d\geq d_0$.  Moreover, let~$\cU$ be as given in that
  lemma. 
  We now invoke Lemma~\ref{lem:PoA_from_good} with~$\epsilon=S^{-2}$.
  Note that condition~\eqref{eq:epsilon_upbdd} does hold, as
  $\epsilon=S^{-2}\leq1/(S-1)=1/(k_{\max}(\cU)-1)$. Combining
  Lemma~\ref{lem:PoA_from_good} and Lemma~\ref{lem:large_value_pack},
  we conclude that~$\PoA(d)\geq\val(\cU)\geq{d/5\log d}$.


\section{Proof of Lemma~\ref{lem:large_value_pack}}
\label{sec:packing_proof}
We shall describe packings in terms of words of certain languages. For
that, we define the languages we are interested in, show the
properties we require, and then prove their existence. Owing to space
limitation, we present only an outline of the proof of
Lemma~\ref{lem:large_value_pack}.


\subsection{Separated families of languages}
\label{sec:words_fam}

Let an integer~$d\geq2$ be fixed.  We consider sets of
words~$L_k\subset[k]^d=\{1,\dots,k\}^d$ for~$k\geq2$.  We refer to
such~$L_k$ as \textit{languages} or $k$-\textit{languages}.  Such
languages~$L_k$ will specify `positions' where we shall
place~$Q_k^d(\epsilon)$ in certain packings (roughly speaking, for
each~$w\in L_k$, we put a certain copy~$Q(w)$
of~$Q_k^d(\epsilon)$ in our packings
(see~\eqref{eq:x^q(j)_def}--\eqref{eq:Q(w)_also} for the definition
of~$Q(w)$)).

We now introduce some conditions on the~$L_k$ that will
help us make sure that we have a packing when we consider the~$Q(w)$ 
($w\in L_k$) all together. 

\begin{definition}[Gapped languages]
  \label{def:gapped_langs}
  Suppose~$k\geq2$ and let a $k$-language $L_k\subset[k]^d$ be given.
  We say that~$L_k$ \textit{misses}~$j$ at coordinate~$i_0$ if every
  word~$w=(w_i)_{1\leq i\leq d}$ in~$L_k$ is such that~$w_{i_0}\neq
  j$.  Furthermore, $L_k$~is said to be \textit{gapped} if, for
  each~$1\leq i\leq d$, either~$L_k$ misses~$k-1$ at~$i$ or~$L_k$
  misses~$k$ at~$i$.
\end{definition}

The reason we are interested in gapped languages is as follows.
Suppose~$L_k$ is a gapped language as in
Definition~\ref{def:gapped_langs}, and
suppose~$w=(w_i)_{1\leq i\leq d}$ and~$w'=(w_i')_{1\leq i\leq d}$ are
distinct words in~$L_k$.  Then~$Q(w)$ and~$Q(w')$ do not overlap (this
can be checked from~\eqref{eq:Q(w)_def} and
Fact~\ref{fact:gap}\ref{enum:fact_gap.ii}; see
Lemma~\ref{lem:indeed_packs}(\textrm{i})).  Thus, if we let $\cP_k$~be
the collection of the~$Q(w)$ ($w\in L_k$), then~$\cP_k$ is a packing.
We now introduce a certain notion of `compatibility' between two
languages~$L_k$ and~$L_{k'}$, so that~$\cP_k$ and~$\cP_{k'}$ can be
put together to obtain a packing if they come from `compatible'
languages~$L_k$ and~$L_{k'}$.

\begin{definition}[Separated languages]
  \label{def:separated_langs}
  Suppose~$2\leq k<k'$ and~$L_k\subset[k]^d$ and~$L_{k'}\subset[k']^d$
  are given.  We say that~$L_k$ and~$L_{k'}$ are \textit{separated}
  if, for any~$w=(w_i)_{1\leq i\leq d}\in L_k$ and
  any~$w'=(w_i')_{1\leq i\leq d}\in L_{k'}$, there is some~$i$ such
  that~$w_i<k<k'=w_i'$.  
\end{definition}

Suppose~$L_k$ and~$L_{k'}$ are gapped and separated.  Consider the
corresponding packings~$\cP_k$ and~$\cP_{k'}$ as above.
Fact~\ref{fact:gap}\ref{enum:fact_gap.i} and~\eqref{eq:Q(w)_def}~imply
that~$\cP_k\cup\cP_{k'}$ is a packing.  To check this,
let~$w=(w_i)_{1\leq i\leq d}\in L_k$ and
any~$w'=(w_i')_{1\leq i\leq d}\in L_{k'}$ be given.  Then, by
definition, there is some~$i$ such that~$w_i<k<k'=w_i'$.  This implies
that~$Q(w)=Q^{(k)}(w)$ and~$Q(w')=Q^{(k')}(w')$ are disjoint `in
the $i$th dimension' (see
Fact~\ref{fact:gap}\ref{enum:fact_gap.i} and
Lemma~\ref{lem:indeed_packs}(\textrm{i})).  

\begin{definition}[Separated families]
  \label{def:separated_families}
  Let~$\cL=(L_k)_{2\leq k\leq S}$ be a family of
  $k$-lan\-guages $L_k\subset[k]^d$. If, for every~$2\leq k<k'\leq S$,
  the languages~$L_k$ and~$L_{k'}$ are separated, then we say
  that~$\cL$ is a \textit{separated} family of languages.
\end{definition}

\begin{remark}
  \label{rem:warm-up} For~$2\leq k\leq d$,
  let~$L_k=\big\{w=(w_i)_{1\leq i\leq k}\in[k]^d\:w_k=k\text{ and
  }w_i<k\text{ for all~$i\neq k$}\big\}$.  One can then check
  that~$\cL=(L_k)_{2\leq k\leq d}$ is a family of gapped, separated
  languages.  Consider the packing~$\cP=\bigcup_{2\leq k\leq d}\cP_k$
  with the~$\cP_k$ defined by the~$L_k$ as above.  We
  have~$\nu_k(\cP)=|L_k|=(k-1)^{d-1}$ (recall~\eqref{eq:nu_k(U)_def})
  and~$\val(\cP)=\sum_{2\leq k\leq d}1/(k-1)\sim\log d$
  (recall~\eqref{eq:val_def}).  The existence of~$\cP$ implies a weak
  form of Theorem~\ref{thm:lwbd_prbsa} (namely, a lower bound
  of~$\Omega(\log d)$ instead of~$\Omega(d/\log d)$); for details, see
  the proof of Theorem~\ref{thm:lwbd_prbsa} in
  Appendix~\ref{sec:pfofThm.lwbd_prbsa}.
\end{remark}

Remark~\ref{rem:warm-up} above illustrates the use we wish to make of
families of gapped, separated languages.  Our focus will soon shift
onto producing much `better' families than the one explicitly defined in
Remark~\ref{rem:warm-up}.  Indeed, the main result in this section is
the following lemma, for which we give a probabilistic proof (see
Section~\ref{sec:pf_languages_lemma} and
Appendix~\ref{sec:languages_lemma}).

\begin{lemma}[Many large, separated gapped languages]
  \label{lem:sep_langs}
  There is an absolute constant~$d_0$ such that, for any~$d\geq d_0$,
  there is a separated family~$\cL=(L_k)_{2\leq k\leq S}$ of gapped
  $k$-languages~$L_k\subset[k]^d$ such that
   \begin{equation}
    \label{eq:sep_langs.Ls}
    |L_k|\geq{10\over11}(k-1)^d,
  \end{equation}
for every~$2\leq k\leq S$, where
  \begin{equation}
    \label{eq:sep_langs.Q}
    S=\llceil2d\over9\log d\rrceil.
  \end{equation}
\end{lemma}

Fix~$\cL=(L_k)_{2\leq k\leq S}$ a family of separated, gapped
$k$-languages~$L_k\subset[k]^d$.  We shall now give, for every
sufficiently small~$\epsilon>0$, the construction of a
packing~$\cU_\epsilon=\cU_\epsilon(\cL)$ of $d$-hypercubes into the
unit bin~$[0,1]^d$ using~$\cL$.  Choosing~$\cL$ suitably, we shall be
able to prove Lemma~\ref{lem:good_exists} below, which takes us very
close to the proof of Lemma~\ref{lem:large_value_pack}.

\subsubsection{The packing~$\cU_\epsilon$.}
\label{sec:packing-cu_epsilon}
The packing~$\cU_\epsilon=\cU_\epsilon(\cL)$
contains copies of the hypercubes~$Q_k^d(\epsilon)$ for~$2\leq k\leq S$.
In fact, for each~$w\in L_k$ ($2\leq k\leq S$), we place a
copy~$Q(w)$ of~$Q_k^d(\epsilon)$ in~$\cU_\epsilon$.
To specify the location of the copy~$Q(w)$ of~$Q_k^d(\epsilon)$
in~$\cU_\epsilon$, we need a definition.

\medskip 

\begin{definition}[Base point coordinates of the~$Q(w)$]
  \label{def:x^k(j)}
  For every~$k\geq2$ and~$0<\epsilon<1/(k-1)$, let
  \begin{equation}
    \label{eq:x^q(j)_def}
    x^{(k)}(j)=x_\epsilon^{(k)}(j)=
    \begin{cases}\displaystyle
      {j-1\over k_-}={(j-1)(1+\epsilon)\over k},
      &\text{if }1\leq j<k\\\displaystyle
      1-{1\over k_-}=1-{1+\epsilon\over k},
      &\text{if }j=k.
    \end{cases}
  \end{equation}
  Moreover, for~$1\leq j\leq k$, let
  \begin{equation}
    \label{eq:y^{(q)}(j)}
    y^{(k)}(j)=x^{(k)}(j)+{1\over k_-}
    =x^{(k)}(j)+{1+\epsilon\over k}.
  \end{equation}
\end{definition}

Note that, for each~$2\leq k\leq S$, we have 
\begin{multline}
  \label{eq:x^d(j).2}
  0=x^{(k)}(1)<y^{(k)}(1)=x^{(k)}(2)<y^{(k)}(2)=x^{(k)}(3)<\cdots
  <y^{(k)}(k-2) \\ 
  =x^{(k)}(k-1)
  <x^{(k)}(k)<y^{(k)}(k-1)<y^{(k)}(k)=1.
\end{multline}
For convenience, for every~$k\geq2$ and every~$1\leq j\leq k$, let
\begin{equation}
  \label{eq:I^(k)(j)_def}
  I^{(k)}(j)=(x^{(k)}(j),y^{(k)}(j))\subset[0,1].
\end{equation}
Now, for each word~${w}=(w_i)_{1\leq i\leq d}\in L_k$ ($2\leq k\leq S$), let
\begin{equation}
  \label{eq:x(w)_def}
  x[w]=x^{(k)}[w]=(x^{(k)}(w_1),\dots,x^{(k)}(w_d))\in\RR^d,
   {\rm and} 
\end{equation}
\begin{equation}
  \label{eq:Q(w)_def}
  Q(w)=Q^{(k)}(w)=x^{(k)}[w]+Q_k^d(\epsilon)\subset[0,1]^d.
\end{equation}
Putting together the definitions, one checks that 
\begin{multline}
 \label{eq:Q(w)_also}
  Q(w)=Q^{(k)}(w)
  =I^{(k)}(w_1)\times\dots\times I^{(k)}(w_d)\\
  =\big(x^{(k)}(w_1),y^{(k)}(w_1)\big)
  \times\dots\times
  \big(x^{(k)}(w_d),y^{(k)}(w_d)\big)
  \subset[0,1]^d.\qquad
\end{multline}

\medskip 

\begin{definition}[Packing~$\cU_\epsilon=\cU_\epsilon(\cL)$]
  \label{def:cU_epsilon(cL)}
  Suppose $\cL=(L_k)_{2\leq k\leq S}$ is a family of separated, gapped
  $k$-languages~$L_k\subset[k]^d$.  Let~$0<\epsilon\leq S^{-2}$.
  Define the packing~$\cU_\epsilon=\cU_\epsilon(\cL)$ as follows.  For
  each~$2\leq k\leq S$ and each~$w\in L_k$, place the
  copy~$Q(w)=Q^{(k)}(w)\subset[0,1]^d$ of~$Q_k^d(\epsilon)$
  in~$\cU_\epsilon$.
\end{definition}

To prove that~$\cU_\epsilon$ is indeed a packing, that is, that the
hypercubes in~$\cU_\epsilon$ are pairwise disjoint, the following
fact can be used (see Appendix~\ref{sec:languages_lemma}).


\begin{fact}
  \label{fact:gap}
  The following assertions hold.
  \begin{enumerate}[label=\rmlabel]
  \item\label{enum:fact_gap.i} Suppose~$2\leq k<k'\leq S$
    and~$0<\epsilon\leq S^{-2}$.  Then
    \begin{equation}
      \label{eq:gap}
      y^{(k)}(k-1)<x^{(k')}(k').
    \end{equation}
    In particular, the intervals~$I^{(k)}(j)$ $(1\leq j<k)$ are disjoint
    from~$I^{(k')}(k')$.
  \item\label{enum:fact_gap.ii} For any~$2\leq k\leq S$, the
    intervals~$I^{(k)}(j)$ $(1\leq j\leq k)$ are pairwise disjoint,
    except for the single pair formed by~$I^{(k)}(k-1)$
    and~$I^{(k)}(k)$.
  \end{enumerate}
\end{fact}

For the next lemma, recall~\eqref{eq:K(U)_def}
and~\eqref{eq:nu_k(U)_def}, and Definition~\ref{def:(1+epsilon)/ZZ}.

\begin{lemma}  
  \label{lem:indeed_packs}
  Suppose $\cL=(L_k)_{2\leq k\leq S}$ is a family of separated,
  non-empty gapped $k$-languages~$L_k\subset[k]^d$.
  Suppose~$0<\epsilon\leq S^{-2}$.
  Let~$\cU_\epsilon=\cU_\epsilon(\cL)$ be the family of all the
  hypercubes~$Q(w)=Q^{(k)}(w)\subset[0,1]^d$ with~$w\in L_k$
  and~$2\leq k\leq S$.  Then the following assertions hold:
  {\rm(i)}~the hypercubes in~$\cU_\epsilon$ are pairwise disjoint and form a
  packing of type~$(1+\epsilon)\ZZ_{\geq2}^{-1}$; 
  {\rm(ii)}~for every~$2\leq k\leq S$, we
  have~$\nu_k(\cU_\epsilon)=|L_k|$; 
  {\rm(iii)}~$|K(\cU_\epsilon)|=S-1$. 
\end{lemma}


\begin{lemma}  
  \label{lem:good_exists}
  There is an absolute constant~$d_0$ for which the following holds
  for any~$d\geq d_0$.  Let $S=\lceil 2d/(9\log d)\rceil.$ \label{eq:S_def.2}
  The unit bin admits a packing~$\cU$ of type $(1+S^{-2})\ZZ_{\geq2}^{-1}$ and
  with~$k_{\max}(\cU)=S$ such that $\val(\cU) \geq(10/11)(S-1)$.
\end{lemma}

Lemma~\ref{lem:good_exists}  is an immediate corollary of
Lemmas~\ref{lem:sep_langs} and~\ref{lem:indeed_packs}. From it, 
the proof of Lemma~\ref{lem:large_value_pack} follows easily:
taking the packing~$\cU$ given in
this lemma, we have that $\val(\cU) \geq(10/11)(S-1) \geq d/(5\log d)$,
as long as $d$ is large enough.

\section{Proof of Lemma~\ref{lem:sep_langs}}
\label{sec:pf_languages_lemma}

We need the following auxiliary fact, which follows from
standard Chernoff bounds for the hypergeometric distribution.

\begin{fact}
  \label{fact:cF}
  There is an absolute constant~$d_0$ such that, for any~$d\geq d_0$,
  there are sets~$F_1,\dots,F_d\subset[d]$ such that {\rm(i)}~for
  every~$1\leq k\leq d$, we have~$|F_k|=\lceil d/2\rceil$ and
  {\rm(ii)}~for every~$1\leq k<k'\leq d$, we
  have~$|F_k\cap F_{k'}|<7d/26$.
\end{fact}


We now proceed to prove Lemma~\ref{lem:sep_langs}.
  Let~$S=\lceil2d/9\log d\rceil$ and let~$F_1,\dots,F_d$ be as in
  Fact~\ref{fact:cF}.  In what follows, we only use~$F_k$ for~$2\leq
  k\leq S$.  For each~$2\leq k\leq S$, we shall
  construct~$L_k\subset[k]^d$ in two parts.  First, let
  \begin{equation}
    \label{eq:L_k'.1}
    L_k'\subset([k]\setminus\{k-1\})^{F_k}
    =\{w=(w_i)_{i\in F_k}:w_i\in[k]\setminus\{k-1\}\text{ for all
    }i\in F_k\}
  \end{equation}
  and then set
  \begin{equation}
    \label{eq:L_k.1}
    \begin{split}
      L_k&=L_k'\times[k-1]^{[d]\setminus F_k}\\
      &=\{w=(w_i)_{1\leq i\leq d}\:
      \exists w'=(w_i')_{i\in F_k}\in L_k' \text{ such that
      }w_i=w_i'\text{ for all }i\in F_k\\
      &\quad\qquad\qquad\qquad\quad\enspace
      \text{ and }w_i\in[k-1]\text{ for all }i\in[d]\setminus F_k\}.
    \end{split}
  \end{equation}
  Note that, by~\eqref{eq:L_k'.1} and~\eqref{eq:L_k.1}, the
  $k$-language~$L_k$ will be gapped ($k-1$ is missed at
  every~$i\in F_k$ and~$k$ is missed at every~$i\in[d]\setminus F_k$).
  We shall prove that there is a suitable choice for the~$L_k'$
  with~$|L_k'|\geq(10/11)(k-1)^d$ ensuring
  that~$\cL=(L_k)_{2\leq k\leq S}$ is separated.  Since we shall then
  have
  \begin{equation}
    \label{eq:L_k.large}
    |L_k|=|L_k'|(k-1)^{d-|F_k|}\geq{10\over11}(k-1)^d,
  \end{equation}
  condition~\eqref{eq:sep_langs.Ls} will be satisfied.  We now proceed
  with the construction of the~$L_k'$.

  Fix~$2\leq k\leq S$.  For~$2\leq\ell<k$, let
  \begin{equation}
    \label{eq:J(ell,k)_def}
    J(\ell,k)=F_k\setminus F_\ell,
  \end{equation}
  and note that
  \begin{equation}
    \label{eq:J_large}
    |J(\ell,k)|>\llceil{d\over2}\rrceil-{7\over26}d
    \geq{3\over13}d.
  \end{equation}
  Let~$v=(v_i)_{i\in F_k}$ be an element
  of~$([k]\setminus\{k-1\})^{F_k}$ chosen uniformly at random.  For
  every~$2\leq\ell<k$, we say that~$v$ is \textit{$\ell$-bad}
  if~$v_i\neq k$ for every~$i\in J(\ell,k)$.  Moreover, we say
  that~$v$ is \textit{bad} if it is $\ell$-bad for
  some~$2\leq\ell<k$.  It is clear that
  \begin{equation}
    \label{eq:PP_v_ell-bad}
    \PP(v\text{ is $\ell$-bad})=\(1-{1\over k-1}\)^{|J(\ell,k)|}
    \leq\ee^{-|J(\ell,k)|/S}
    \leq\exp\(-{3d\over13\lceil2d/9\log d\rceil}\)
    \leq d^{-1},
  \end{equation}
  for every large enough~$d$, whence
  \begin{equation}
    \label{eq:PP_v_bad}
    \PP(v\text{ is bad})\leq Sd^{-1}\leq{1\over4\log d}\leq{1\over11}
  \end{equation}
  if~$d$ is large enough.  Therefore, at least~$(10/11)(k-1)^{|F_k|}$
  words~$v\in([k]\setminus\{k-1\})^{F_k}$ are not bad, as long as~$d$
  is large enough.  We let~$L_k'\subset([k]\setminus\{k-1\})^{F_k}$ be
  the set of such good words.  The following claim completes the proof
  of Lemma~\ref{lem:sep_langs}.

  \begin{claim}
    \label{claim:L_k'_good}
    With the above choice of~$L_k'$ $(2\leq k\leq S)$, the
    family~$\cL=(L_k)_{2\leq k\leq S}$ of the languages~$L_k$ as
    defined in~\eqref{eq:L_k.1} is separated.
  \end{claim}
  \begin{proof}
    Fix~$2\leq\ell<k\leq S$.  We show that~$L_\ell$ and~$L_k$ are
    separated.  Let~$u=(u_i)_{1\leq i\leq d}\in L_\ell$
    and~$w=(w_i)_{1\leq i\leq d}\in L_k$ be given.  By the definition
    of~$L_k$, there is~$v=(v_i)_{i\in F_k}\in L_k'$ such
    that~$w_i=v_i$ for all~$i\in F_k$.  Furthermore, since~$v\in L_k'$
    is not a bad word, it is not $\ell$-bad.  Therefore, there
    is~$i_0\in J(\ell,k)=F_k\setminus F_\ell$ for which we
    have~$v_{i_0}=k$.  Observing that~$i_0\notin F_\ell$ and recalling
    the definition of~$L_\ell$, we see that $u_{i_0}<\ell<k=v_{i_0}=w_{i_0}$,
    as required.
  \end{proof}




\bibliographystyle{splncs}
\bibliography{refs,mybibs,new}

\def\cprime{$'$}
\begin{thebibliography}{10}

\bibitem{CoffmanGJ97}
{Coffman, Jr.}, E.G., Garey, M.R., Johnson, D.S.:
\newblock Approximation algorithms for bin packing: a survey.
\newblock In Hochbaum, D., ed.: Approximation Algorithms for {NP}-hard
  Problems.
\newblock PWS (1997)  46--93

\bibitem{Vliet92}
{van Vliet}, A.:
\newblock An improved lower bound for online bin packing algorithms.
\newblock Inform. Process. Lett. \textbf{43} (1992)  277--284

\bibitem{Seiden02}
Seiden, S.S.:
\newblock On the online bin packing problem.
\newblock Journal of the Association for Computing Machinery \textbf{49}(5)
  (2002)  640--671

\bibitem{BaloghBG10}
Balogh, J., B{\'{e}}k{\'{e}}si, J., Galambos, G.:
\newblock New lower bounds for certain classes of bin packing algorithms.
\newblock In: Proc. of the 8th {I}nternational {W}orkshop on {A}pproximation
  and {O}nline {A}lgorithms. (2010)  25--36

\bibitem{HeydrichS2016a}
Heydrich, S., van Stee, R.:
\newblock Improved lower bounds for online hypercube packing.
\newblock CoRR \textbf{abs/1607.01229} (2016)

\bibitem{HeydrichS16b}
Heydrich, S., van Stee, R.:
\newblock {Beating the Harmonic Lower Bound for Online Bin Packing}.
\newblock In: ICALP 2016. Volume~55 of LIPIcs., Dagstuhl, Germany (2016)
  41:1--41:14 (See newer version http://arxiv.org/abs/1511.00876v5).

\bibitem{BaloghBDEL17arx-a}
Balogh, J., B{\'{e}}k{\'{e}}si, J., D{\'{o}}sa, G., Epstein, L., Levin, A.:
\newblock Lower bounds for several online variants of bin packing.
\newblock CoRR \textbf{abs/1708.03228} (2017)

\bibitem{Csirik89}
Csirik, J.:
\newblock An on-line algorithm for variable-sized bin packing.
\newblock Acta Informatica \textbf{26}(8) (1989)  697--709

\bibitem{Seiden01}
Seiden, S.S.:
\newblock An optimal online algorithm for bounded space variable-sized bin
  packing.
\newblock SIAM J. Discrete Math. \textbf{14} (2001)  458--470

\bibitem{CsirikV93}
Csirik, J., {van Vliet}, A.:
\newblock An on-line algorithm for multidimensional bin packing.
\newblock Operations Research Letters \textbf{13} (1993)  149--158

\bibitem{EpsteinS05}
Epstein, L., Stee, R.v.:
\newblock Optimal online algorithms for multidimensional packing problems.
\newblock SIAM J. Comput. \textbf{35}(2) (2005)  431--448

\bibitem{EpsteinS07}
Epstein, L., {van Stee}, R.:
\newblock Bounds for online bounded space hypercube packing.
\newblock Discrete Optimization \textbf{4}(2) (2007)  185--197

\bibitem{Nash51}
Nash, J.:
\newblock Non-cooperative games.
\newblock Annals of Mathematics \textbf{54}(2) (1951)  286--295

\bibitem{Aumann59}
Aumann, R.J.:
\newblock Acceptable points in general cooperative n-person games.
\newblock In Luce, R.D., Tucker, A.W., eds.: Annals of Mathematical Study.
  Volume~40.
\newblock University Press (1959)  287--324

\bibitem{KoutsoupiasP99}
Koutsoupias, E., Papadimitriou, C.H.:
\newblock Worst-case equilibria.
\newblock In: Proc. of the 16th {A}nnual {S}ymposium on {T}heoretical {A}spects
  of {C}omputer {S}cience. (1999)  404--413

\bibitem{Bilo06}
Bil\`o, V.:
\newblock On the packing of selfish items.
\newblock In: Proc. 20th {I}nternacional {P}arallel and {D}istributed
  {P}rocessing {S}ymposium, IEEE (2006)  9--18

\bibitem{YuZ08}
Yu, G., Zhang, G.:
\newblock Bin packing of selfish items.
\newblock In: Proc. of the 4th {I}nternational {W}orkshop on {I}nternet and
  {N}etwork {E}conomics. (2008)  446--453

\bibitem{EpsteinK11}
Epstein, L., Kleiman, E.:
\newblock Selfish bin packing.
\newblock Algorithmica \textbf{60}(2) (2011)  368--394

\bibitem{EpsteinKM16}
Epstein, L., Kleiman, E., Mestre, J.:
\newblock Parametric packing of selfish items and the subset sum algorithm.
\newblock Algorithmica \textbf{74}(1) (January 2016)  177--207

\bibitem{MaDHTYZ13}
Ma, R., D\'osa, G., Han, X., Ting, H.F., Ye, D., Zhang, Y.:
\newblock A note on a selfish bin packing problem.
\newblock Journal of Global Optimization \textbf{56}(4) (2013)  1457--1462

\bibitem{FernandesFMW12}
Fernandes, C.G., Ferreira, C.E., Miyazawa, F.K., Wakabayashi, Y.:
\newblock Selfish square packing.
\newblock In: Proc. of the {VI} {L}atin-american {A}lgorithms, {G}raphs and
  {O}ptimization {S}ymposium. Volume~37 of Elect. Notes in Discrete
  Mathematics. (2011)  369--374

\bibitem{FernandesFMW17}
Fernandes, C.G., Ferreira, C.E., Miyazawa, F.K., Wakabayashi, Y.:
\newblock Prices of anarchy of selfish {2D} bin packing games.
\newblock CoRR \textbf{abs/1707.07882} (2017)

\bibitem{Epstein13}
Epstein, L.:
\newblock Bin packing games with selfish items.
\newblock In: Proc.\ of the 38th {I}nternational {S}ymposium on {M}athematical
  {F}oundations of {C}omputer {S}cience. (2013)  8--21

\bibitem{janson00:_random_graph}
Janson, S., {\L}uczak, T., Ruci{\'n}ski, A.:
\newblock Random graphs.
\newblock Wiley-Interscience, New York (2000)

\bibitem{MeirM68}
Meir, A., Moser, L.:
\newblock On packing of squares and cubes.
\newblock J. Combinatorial Theory Ser. A \textbf{5} (1968)  116--127

\end{thebibliography}



\appendix

\section*{Appendix. Omitted proofs}
\label{sec:omitted-proofs}

\section{Proof of Theorem~\ref{thm:lwbd_prbsa}}
\label{sec:pfofThm.lwbd_prbsa}
Let~$\cA$ be any algorithm for the online bounded space $d$-hypercube
bin packing problem.  Let~$M$ be the maximum
number of bins that~$\cA$ is allowed to leave open during its
execution. To prove that~$\cA$ has asymptotic performance
ratio~$\Omega(d/\log d)$, we construct a
suitable instance~$\cI$ for~$\cA$.

Let a packing~$\cU$ as in the statement of
Lemma~\ref{lem:large_value_pack} be fixed.  The instance~$\cI$ will be
constructed by choosing a suitable integer~$N$ and then arranging the
hypercubes in~$2MN$ copies of~$\cU$ in a linear order, with all the
hypercubes of the same size appearing together.  Let us now formally
describe~$\cI$.

Let
\begin{equation}
  \label{eq:N_def}
  N=\prod_{k\in K(\cU)}(k-1)^d
\end{equation}
and, for every~$k\in K(\cU)$, let
\begin{equation}
  \label{eq:N_sigma_def}
  N_{\hat k}
  ={N\over(k-1)^d}
  =\prod_{k'\in K(\cU)\setminus\{k\}}(k'-1)^d.
\end{equation}
Recall that~$\cU$ contains~$\nu_k(\cU)$ copies of~$Q_k^d(\epsilon)$
for every~$k\in K(\cU)$.  Let~$K=|K(\cU)|$ and
suppose~$K(\cU)=\{k_1,\dots,k_K\}$.  The instance~$\cI$ that we shall
construct is the concatenation of~$K$ segments,
say~$\cI=\cI_1\dots\cI_K$,
with each segment~$\cI_\ell$ ($1\leq\ell\leq K$) composed of a certain
number of copies of~$Q_{k_\ell}^d(\epsilon)$.  For
every~$1\leq\ell\leq K$, set
\begin{equation}
  \label{eq:f(ell))_def}
  f(\ell)=2MN\nu_{k_\ell}(\cU),
\end{equation}
and
\begin{equation}
  \label{eq:cI_ell_def}
  \cI_\ell=(Q_{k_\ell}^d(\epsilon),\dots,Q_{k_\ell}^d(\epsilon))
  =Q_{k_\ell}^d(\epsilon)^{f(\ell)}.
\end{equation}
That is, $\cI_\ell$ is composed of a sequence of~$f(\ell)$ copies
of~$Q_{k_\ell}^d(\epsilon)$.  This completes the definition of our
instance~$\cI$.

Let us first state the following fact concerning the offline packing
of the hypercubes in~$\cI$.  This fact is clear, as we obtained~$\cI$ by
rearranging the hypercubes in~$2MN$ copies of~$\cU$.

\begin{fact}
  \label{fact:cI_packing}
  The hypercubes in~$\cI$ can be packed into at most~$2MN$
  unit bins.
\end{fact}

We now prove that, when~$\cA$ is given the instance~$\cI$ above, it
will have performance ratio at least as bad as~$\val(\cU)/2$.
In view of~\eqref{eq:large_value_pack} in
Lemma~\ref{lem:large_value_pack}, this will complete the proof of
Theorem~\ref{thm:lwbd_prbsa}.

Let us examine the behaviour of~$\cA$ when given input~$\cI$.
Fix~$1\leq\ell\leq K$ and suppose~$\cA$ has already seen the hypercubes
in~$\cI_1\dots\cI_{\ell-1}$ and it has already packed them somehow.
We now consider what happens when~$\cA$ examines the~$f(\ell)$ hypercubes
in~$\cI_\ell$, which are all copies of~$Q_{k_\ell}^d(\epsilon)$.

Clearly, since~$\epsilon>0$, the~$f(\ell)$ copies
of~$Q_{k_\ell}^d(\epsilon)$ in~$\cI_\ell$ cannot be packed into fewer
than
\begin{equation}
  \label{eq:cI_ell_lwbd}
  {f(\ell)\over(k_\ell-1)^d}
  ={2MN\nu_{k_\ell}(\cU)\over(k_\ell-1)^d}
  =2MN_{\hat k}\nu_{k_\ell}(\cU)
  \geq MN_{\hat k}\nu_{k_\ell}(\cU)+M
\end{equation}
unit bins.  Therefore, even if some hypercubes in~$\cI_\ell$ are placed in
bins still left open after the processing of~$\cI_1\dots\cI_{\ell-1}$,
the hypercubes in~$\cI_\ell$ will add at
least~$MN_{\hat k}\nu_{k_\ell}(\cU)$ new bins to the output of~$\cA$.
Thus, the total number of bins that~$\cA$ will use when
processing~$\cI$ is at least
\begin{equation}
  \label{eq:cA_lwbd}
  \sum_{k\in K(\cU)}MN_{\hat k}\nu_k(\cU)
  =MN\sum_{k\in K(\cU)}(k-1)^{-d}\nu_k(\cU)
  =MN\val(\cU).
\end{equation}
In view of Fact~\ref{fact:cI_packing}, it follows that the asymptotic
performance ratio of~$\cA$ is at least
\begin{equation}
  \label{eq:cA_apr_lwbd}
  {MN\val(\cU)\over2MN}={1\over2}\val(\cU),
\end{equation}
as required.  This completes the proof of
Theorem~\ref{thm:lwbd_prbsa}. 


\section{Proofs of Section~\ref{sec:pfofThm.PoA_lwbd}}
\label{app:pfofThm.PoA_lwbd}

We start with the proof of Lemma~\ref{lem:an_equilibrium}, which
depends on the following simple result.

\begin{proposition}
  \label{prop:an_ineq}
  Suppose~$d\geq2$, $\ell\geq k+1$ and~$k>1$.  Then
  \begin{equation}
    \label{eq:an_ineq}
    \(1-{1\over k}\)^d+{1\over\ell^d}<\(1-{1\over\ell}\)^d.
  \end{equation}
\end{proposition}
\begin{proof}
  Inequality~\eqref{eq:an_ineq} is equivalent to
  \begin{equation}
    \label{eq:an_ineq.2}
    \ell^d(k-1)^d+k^d<(\ell-1)^dk^d.
  \end{equation}
  That is, 
  \begin{equation}
    \label{eq:an_ineq.3}
    \big(\ell^d(k-1)^d+k^d\big)^{1/d}<(\ell-1)k.
  \end{equation}
  Since~$d\geq2$, we have
  \begin{equation}
    \label{eq:an_ineq.4}
    \big(\ell^d(k-1)^d+k^d\big)^{1/d}\leq\(\ell^2(k-1)^2+k^2\)^{1/2}.
  \end{equation}
  Therefore, it suffices to prove that~\eqref{eq:an_ineq.3} holds
  for~$d=2$, that is, 
  \begin{equation}
    \label{eq:an_ineq.5}
    \ell^2(k-1)^2+k^2<(\ell-1)^2k^2.
  \end{equation}
  A quick calculation shows that~\eqref{eq:an_ineq.5} is equivalent to
  \begin{equation}
    \label{eq:an_ineq.6}
    \(2k-1\)\ell > 2k^2.
  \end{equation}
  Since~$2<k+1\leq\ell$, we conclude that   (\ref{eq:an_ineq.5})
  holds, and this completes the proof.
\end{proof}

\medskip

\begin{proof}{\bf [Proof of Lemma~\ref{lem:an_equilibrium}]}
  Recall that, for any integer~$q$, we
  let~$q_-=q_{-\epsilon}=q/(1+\epsilon)$.  Let~$2\leq k<\ell$ be
  integers and suppose~$\cP$ includes a packing~$\cU(k)$ of
  type~$\cH_k^d(\epsilon)$ and a packing~$\cU(\ell)$ of
  type~$\cH_\ell^d(\epsilon)$.  It suffices to show that the cost of a
  copy of~$Q_\ell^d(\epsilon)$ within~$\cU(\ell)$ is smaller than the
  cost that it would incur if it were moved into~$\cU(k)$.  Recall
  that the volume of~$Q_k^d(\epsilon)$
  is~$(1+\epsilon)^d/k^d=1/k_-^d$.  Thus, we have to show the
  following inequality:
  \begin{equation}
    \label{eq:Nash_ineq}
    {1/\ell_-^d\over(k-1)^d/k_-^d+1/\ell_-^d}
    >{1/\ell_-^d\over(\ell-1)^d/\ell_-^d},
  \end{equation}
  which is equivalent to
  \begin{equation}
    \label{eq:Nash_ineq.2}
    \(k-1\over k_-\)^d+{1\over\ell_-^d}<\(\ell-1\over\ell_-\)^d,
  \end{equation}
  which, in turn, is equivalent to
  \begin{equation}
    \label{eq:Nash_ineq.3}
    \(k-1\over k\)^d+{1\over\ell^d}<\(\ell-1\over\ell\)^d
  \end{equation}
  (recall that~$q_-=q/(1+\epsilon)$).
    Inequality~\eqref{eq:Nash_ineq.3} is precisely
    inequality~\eqref{eq:an_ineq} asserted in Proposition~\ref{prop:an_ineq}.
\end{proof}

\medskip

\begin{proof}\textbf{[Proof of Lemma~\ref{lem:PoA_from_good}]}
  Let $\cU$ be a packing as in the statement of
  Lemma~\ref{lem:PoA_from_good}. 
  We shall again use the quantities~$N$ and~$N_{\hat k}$ as defined
  in~\eqref{eq:N_def} and~\eqref{eq:N_sigma_def}.
  Consider the packing~$\cP=(\cU_1,\dots,\cU_N)$, where each~$\cU_n$
  ($1\leq n\leq N$) is a copy of~$\cU$.  Fix~$k\in K(\cU)$.
  Let~$\nu_k(\cP)$ be the total number of copies of~$Q_k^d(\epsilon)$
  in~$\cP$.  Then~$\nu_k(\cP)=N\nu_k(\cU)=N_{\hat k}(k-1)^d\nu_k(\cU)$.
  In view of~\eqref{eq:epsilon_upbdd} (recall the observation just
  after Definition~\ref{def:cH}), the~$\nu_k(\cP)$ copies
  of~$Q_k^d(\epsilon)$ in~$\cP$ may be arranged
  into~$\nu_k(\cP)/(k-1)^d=N_{\hat k}\nu_k(\cU)$ copies of packings of
  type~$\cH_k^d(\epsilon)$.  Doing this for every~$k\in K(\cU)$, and
  taking the resulting collection of packings of
  type~$\cH_k^d(\epsilon)$ ($k\in K(\cU)$), we obtain a packing~$\cP'$
  of all the hypercubes in~$\cP$.  Clearly,
  \begin{enumerate}[label=\rmlabel]
  \item\label{enum:cP'.i} $\cP'$ is a packing of
    type~$\cH^d(\epsilon)$ and
  \item\label{enum:cP'.ii} the number of bins in~$\cP'$ is
    \begin{equation}
      \label{eq:cP'_size}
      |\cP'| = \sum_{k\in K(\cU)}N_{\hat k}\nu_k(\cU)
      =N\sum_{k\in K(\cU)}(k-1)^{-d}\nu_k(\cU)
      =N\val(\cU).
    \end{equation}
  \end{enumerate}
  Lemma~\ref{lem:an_equilibrium} and~\ref{enum:cP'.i} tell us
  that~$\cP'$ is a Nash equilibrium.  On the other hand, the fact
  that~$\cP$ uses~$N$ bins and~\ref{enum:cP'.ii} tell us that
  \begin{equation}
    \label{eq:PoA_lwbd_instance}
    \PoA(d)\geq{|\cP'|\over|\cP|}={N\val(\cU)\over N}=\val(\cU),
  \end{equation}
  as required. 
\end{proof}


\bigskip

\section{Proof of Theorem~\ref{thm:SPoA_lwbd} stated in
  Section~\ref{sec:introduction}} 
\label{sec:pfofThm.strong_PoA_lwbd}

Theorem~\ref{thm:SPoA_lwbd} follows from
Lemma~\ref{lem:large_value_pack2+} and the next two lemmas.
These lemmas are the  analogue of Lemmas~\ref{lem:an_equilibrium}
and~\ref{lem:PoA_from_good} for the strong price of anarchy. 

\begin{definition}[Packings of type~$\cH_{2^+}^d(\epsilon)$] 
  \label{def:cH2}
  A packing~$\cP=(\cU_1,\cU_2,\dots)$ is said to be of
  \textit{type~$\cH_{2^+}^d(\epsilon)$} if for each~$i$ there is some
  integer~$k\geq1$ such that~$\cU_i$ is a packing of
  type~$\cH_{2^k}^d(\epsilon)$.
\end{definition}

\begin{lemma} 
  \label{lem:a_S_equilibrium}
  Let~$d\geq2$ and~$\epsilon>0$ be given.  Any
  packing~$\cP=(\cU_1,\cU_2,\dots)$ of
  \textit{type~$\cH_{2^+}^d(\epsilon)$} is a strong Nash equilibrium.
\end{lemma}
\begin{proof} 
  The proof is by contradiction. Suppose there is a packing $\cP$ of
  type $\cH_{2^+}^d(\epsilon)$ that is not a strong
  Nash equilibrium. 
  Recall that $\cP$ is a homogeneous packing and each bin
  in $\cP$ has the maximum number of identical copies
  of~$Q_k^d(\epsilon)$, where $k=2^t$, for some $t$.
  Let $S$ be a set of hypercubes (coalition) that
  can migrate to another bin, say $B'$,  decreasing the cost of each of its
  items.
  Let $h$ be a smallest hypercube in $S$. Suppose $h$ is in bin $B$,
  and has side length $s(h)$.  Since $h$ can migrate to bin $B'$, it
  follows that $s(h)$ is smaller than (and also divides) the side
  length $s(h')$ of any item $h'$ in $B'$. Thus, we can replace each
  hypercube $h'$ originally in $B'$ by $(s(h')/s(h))^d$ hypercubes of
  side length $s(h)$, with the same total volume. Likewise, each
  hypercube $h'\in S$ that migrates to $B'$ can be replaced by
  $(s(h')/s(h))^d$ hypercubes of side length $s(h)$.  After this
  replacement, the new equivalent packing configuration of bin $B'$
  has only hypercubes of side length $s(h)$, and has an occupied
  volume larger than the occupied volume of bin $B$ (before the
  migration), a contradiction, because bin $B$ had the best possible
  occupied volume with items of side length $s(h)$.
\end{proof}
   

\begin{lemma}  
  \label{lem:SPoA_from_good2+}
  If~$\cU$ is a packing of $d$-hypercubes into a unit bin of type
  $(1+\epsilon)\ZZ_{2^+}^{-1}$, where
  \begin{equation}
    \label{eq:strong_epsilon_upbdd}
    0<\epsilon\leq{1\over k_{\max}(\cU)-1},
  \end{equation}
  then~$\SPoA(d)\geq\val(\cU)$.
\end{lemma}

Lemma~\ref{lem:SPoA_from_good2+} follows from
Lemma~\ref{lem:a_S_equilibrium}, in the same way that
Lemma~\ref{lem:PoA_from_good} follows from
Lemma~\ref{lem:an_equilibrium}.


The proof of Theorem~\ref{thm:SPoA_lwbd} follows from
Lemmas~\ref{lem:large_value_pack2+} and~\ref{lem:SPoA_from_good2+} in
the same way that Theorem~\ref{thm:PoA_lwbd} follows from
Lemmas~\ref{lem:large_value_pack} and~\ref{lem:PoA_from_good}.


\section{Proofs of Section~\ref{sec:packing_proof}  
(Proof of Lemma~\ref{lem:large_value_pack})}
\label{sec:languages_lemma}

\begin{proof} \textbf{[Proof of Fact~\ref{fact:gap}]}
  Assertion~\ref{enum:fact_gap.ii} is clear
  (recall~\eqref{eq:x^d(j).2}).  The second assertion
  in~\ref{enum:fact_gap.i} follows from inequality~\eqref{eq:gap}, and
  therefore it suffices to verify that inequality.  We
  have~$y^{(k)}(k-1)=x^{(k)}(k-1)+1/k_-=(k-1)/k_-=1+\epsilon-(1+\epsilon)/k$.
  Moreover, $x^{(k')}(k')=1-1/(k')_-=1-(1+\epsilon)/k'$.  Therefore,
  \eqref{eq:gap}~is equivalent to
  \begin{equation}
    \label{eq:gap.1}
    \epsilon<(1+\epsilon)\({1\over k}-{1\over k'}\).
  \end{equation}
  Since~$k+1\leq k'\leq S$ and~$\epsilon\leq S^{-2}$,
  inequality~\eqref{eq:gap.1} does hold.
\end{proof}


\medskip


\begin{proof}\textbf{[Proof of Lemma~\ref{lem:indeed_packs}]}
  Let us check that the~$Q(w)$ in~$\cU_\epsilon$ are indeed pairwise
  disjoint.  Let~$w=(w_i)_{1\leq i\leq d}\in L_k$
  and~$w'=(w_i')_{1\leq i\leq d}\in L_{k'}$
  with~$2\leq k\leq k'\leq S$ with~$w\neq w'$ be given, and
  consider~$Q(w)=Q^{(k)}(w)$ and~$Q(w')=Q^{(k')}(w')$.  We have to
  show that
  \begin{equation}
    \label{eq:Qs_disjoint}
    Q(w)\cap Q(w')=\emptyset.
  \end{equation}

  Suppose first that~$k=k'$.  In that case, both~$w$ and~$w'$ are
  in~$L_k=L_{k'}$.  Since~$w\neq w'$, there is some~$1\leq i\leq d$
  such that~$w_i\neq w_i'$.  Furthermore, since~$L_k$ is gapped,
  either~$k-1$ or~$k$ is missed by~$L_k$ at~$i$.  In particular, the
  pair~$\{w_i,w_i'\}$ cannot be the pair~$\{k-1,k\}$ and
  therefore
  \begin{equation}
    \label{eq:disjoint_Ii}
    I^{(k)}(w_i)\cap I^{(k)}(w_i')=\emptyset    
  \end{equation}
  (recall Fact~\ref{fact:gap}\ref{enum:fact_gap.ii}).
  Expression~\eqref{eq:Q(w)_also} applied to~$Q(w)$ and~$Q(w')$,
  together with~\eqref{eq:disjoint_Ii},
  confirms~\eqref{eq:Qs_disjoint} when~$k=k'$.  

  Suppose now that~$k<k'$.  Since~$L_k$ and~$L_{k'}$ are separated,
  there is some~$1\leq i_0\leq d$ such that~$w_{i_0}<k<k'=w_{i_0}'$.
  Fact~\ref{fact:gap}\ref{enum:fact_gap.i} tells us that
  \begin{equation}
    \label{eq:disjoint_Ii0}
    I^{(k)}(w_{i_0})\cap I^{(k')}(w_{i_0}')=\emptyset.
  \end{equation}
  Expression~\eqref{eq:Q(w)_also} applied to~$Q(w)$ and~$Q(w')$,
  together with~\eqref{eq:disjoint_Ii0},
  confirms~\eqref{eq:Qs_disjoint} in this case also.  We therefore
  conclude that~$\cU_\epsilon$ is indeed a packing.

  The hypercubes in~$\cU_\epsilon$ are copies of the
  hypercubes~$Q_k^d(\epsilon)$ for~$2\leq k\leq S$, and
  therefore~$\cU_\epsilon$ is a packing of
  type~$(1+\epsilon)\ZZ_{\geq2}^{-1}$.  This concludes the proof of
  assertion~\ref{lem:indeed_packs}(\textrm{i}).
  Assertions~\ref{lem:indeed_packs}(\textrm{ii})
  and~\ref{lem:indeed_packs}(\textrm{iii}) are clear. 
\end{proof}

Lemma~\ref{lem:good_exists} is an immediate corollary of
Lemmas~\ref{lem:sep_langs} and~\ref{lem:indeed_packs}.

\begin{proof}\textbf{[Proof of Lemma~\ref{lem:large_value_pack}]}
  Let~$\cU$ be as given in Lemma~\ref{lem:good_exists}.  We claim
  that~$\cU$ will do.  In fact, 
\begin{equation}
    \label{eq:main_thm.pf}
    \val(\cU)\geq {10\over11}\(\llceil2d\over9\log d\rrceil-1\)
    \geq{d\over5\log d},
  \end{equation}
  where the last inequality holds as long as~$d$ is large enough.
  This completes the proof of Lemma~\ref{lem:large_value_pack}.
\end{proof}


\section{Proof of Fact~\ref{fact:cF}} 
\label{sec:proof-of-set-system-fact}
  Let~$r=\lceil d/2\rceil$.  We select each~$F_k$ ($1\leq k\leq d$)
  among the $r$-element subsets of~$[d]$ uniformly at random, with
  each choice independent of all others.  Let~$s=7d/26$.
  Note that, for any~$k\neq k'$, we have~$\EE(|F_k\cap
  F_{k'}|)=r^2/d$.  Let~$\lambda=r^2/d$.  Let
  \begin{equation}
    \label{eq:t_from_s}
    t=s-\lambda\geq s-(d/2+1)^2/d
    \geq{7d\over26}-{1\over d}\({d^2\over4}+d+1\)
    \geq{d\over52}-2\geq{d\over53},
  \end{equation}
  as long as~$d$ is large enough.  We may now apply a Chernoff bound
  for the hypergeometric distribution (see, for example,
  \cite[Theorem~2.10, inequality~(2.12)]{janson00:_random_graph}) to
  see that
  \begin{equation}
    \label{eq:hyper_dev}
    \PP(|F_k\cap F_{k'}|\geq s)
    =\PP(|F_k\cap F_{k'}|\geq\lambda+t)
    \leq\exp\(-{2(d/53)^2\over\lceil d/2\rceil}\)
    \leq\ee^{-{3d/53^2}}
  \end{equation}
  for every large enough~$d$.  Therefore, the expected number of
  pairs~$\{k,k'\}$ ($1\leq k<k'\leq d$) for which~$|F_k\cap
  F_{k'}|\geq s$ is less than~$d^2\exp(-3d/53^2)$, 
  which tends to~$0$ as~$d\to\infty$.  Therefore, for any large
  enough~$d$, a family of sets~$F_1,\dots,F_d$ as required does exist.


\section{Proof of Lemma 2}

\subsubsection{The packing~$\cU_{\epsilon,2^+}$.}
\label{sec:packing-cu_epsilon2+}
The construction of~$\cU_{\epsilon,2^+}=\cU_{\epsilon,2^+}(\cL)$ will
be based on a variant of Lemma~\ref{lem:sep_langs} (namely,
Lemma~\ref{lem:sep_langs2+}), to be stated in a short while.  Let
\begin{equation}
  \label{eq:S'_def}
  S'=\llceil\log_2d-\log_2\log d-3\rrceil.
\end{equation}
For~$2\leq k\leq S'$, let~$t(k)=2^{k-1}$.
Moreover, let~$T(S')=\{t(k)\:2\leq k\leq
S'\}=\{2,2^2,\dots,2^{S'-1}\}$. 

\begin{lemma}[Many large, separated gapped languages (variant)]
  \label{lem:sep_langs2+}
  There is an absolute constant~$d_0$ such that, for any~$d\geq d_0$,
  there is a separated family~$\cL=(L_t)_{t\in T(S')}$ of gapped
  $t$-languages~$L_t\subset[t]^d$, where~$S'$ is as
  in~\eqref{eq:S'_def}, and
  \begin{equation}
    \label{eq:sep_langs.Ls+}
    |L_t|\geq{10\over11}(t-1)^d
  \end{equation}
  for every~$t\in T(S')$.
\end{lemma}

The proof of Lemma~\ref{lem:sep_langs2+} is very similar to the
proof of Lemma~\ref{lem:sep_langs}, and is omitted here.  With
Lemma~\ref{lem:sep_langs2+} at hand, we may define the
packing~$\cU_{\epsilon,2^+}$.  Let languages~$L_t$ ($t\in T(S')$) as
in Lemma~\ref{lem:sep_langs2+} be fixed.  For each~$w\in L_t$, we
consider~$x[w]=x^{(t)}[w]$ as defined in~\eqref{eq:x(w)_def}, namely,
\begin{equation}
  \label{eq:x(w)_def2}
  x[w]=x^{(t)}[w]=(x^{(t)}(w_1),\dots,x^{(t)}(w_d))\in\RR^d.
\end{equation}
Furthermore, we consider~$Q(w)=Q^{(t)}(w)$ as defined
in~\eqref{eq:Q(w)_def}, namely,
\begin{equation}
  \label{eq:Q(w)_def.2}
  Q(w)=Q^{(t)}(w)=x^{(t)}(w)+Q_t^d(\epsilon)\subset[0,1]^d.
\end{equation}
We now define the packing~$\cU_{\epsilon,2^+}$.

\begin{definition}[Packing~$\cU_{\epsilon,2^+}=\cU_{\epsilon,2^+}(\cL)$]
  \label{def:cU_epsilon2+(cL)}
  Suppose $\cL=(L_t)_{t\in T(S')}$ is a family of separated, gapped
  $t$-languages~$L_t\subset[t]^d$.  Let~$0<\epsilon\leq2^{-2(S'-1)}$.
  Define the packing~$\cU_{\epsilon,2^+}=\cU_{\epsilon,2^+}(\cL)$ as
  follows.  For each~$t\in T(S')$ and each~$w\in L_t$, place the
  copy~$Q(w)=Q^{(t)}(w)\subset[0,1]^d$ of~$Q_t^d(\epsilon)$
  in~$\cU_{\epsilon,2^+}$.
\end{definition}

We now state, without proof, the analogue of 
Lemma~\ref{lem:indeed_packs} for~$\cU_{\epsilon,2^+}$

\begin{lemma}
  \label{lem:indeed_packs2+}
  Suppose $\cL=(L_t)_{t\in T(S')}$ is a family of separated, non-empty
  gapped $t$-languages~$L_t\subset[t]^d$.
  Suppose~$0<\epsilon\leq2^{-2(S'-1)}$.
  Let~$\cU_{\epsilon,2^+}=\cU_{\epsilon,2^+}(\cL)$ be the family of
  all the hypercubes~$Q(w)=Q^{(t)}(w)\subset[0,1]^d$ with~$w\in L_t$
  and~$t\in T(S')$.  Then the following assertions hold.
  \begin{enumerate}[label=\rmlabel]
  \item\label{enum:indeed.i2+} The hypercubes in~$\cU_{\epsilon,2^+}$ are
    pairwise disjoint and form a packing of
    type $(1+\epsilon)\ZZ_{2^+}^{-1}$. 
  \item\label{enum:indeed.ii2+} For every~$t\in T(S')$, we
    have~$\nu_t(\cU_{\epsilon,2^+})=|L_t|$.
  \item\label{enum:indeed.iii2+} We
    have~$|K(\cU_{\epsilon,2^+})|=S'-1$.
  \end{enumerate}
\end{lemma}

The following result 
is an immediate corollary of
Lemmas~\ref{lem:sep_langs2+} and~\ref{lem:indeed_packs2+}.

\begin{lemma}
  \label{lem:good_exists2+}
  There is an absolute constant~$d_0$ for which the following holds
  for any~$d\geq d_0$.  Let
  \begin{equation}
    \label{eq:S'_def.2}
    S'=\llceil\log_2d-\log_2\log d-3\rrceil
  \end{equation}
  and~$\epsilon=2^{-2(S'-1)}$.  The unit bin admits an
  $(S'-1,10/11)$-good packing~$\cU$ of type
  $(1+\epsilon)\ZZ_{2^+}^{-1}$ and with~$k_{\max}(\cU)=2^{S'-1}$.
\end{lemma}

The proof of Lemma~\ref{lem:large_value_pack2+} follows from
Lemma~\ref{lem:good_exists2+} in the same way that the proof of 
Lemma~\ref{lem:large_value_pack} follows from
Lemma~\ref{lem:good_exists}.


\section{Upper bound for the prices of anarchy}
\label{sec:further}
It is not difficult to obtain a simple upper bound for the price of
anarchy of the selfish hypercube bin packing game. Such a bound can be
obtained using the following result.

\begin{theorem}[Meir and Moser \cite{MeirM68}]\label{NewRef1} 
  Every set $S$ of $d$-hypercubes whose largest hypercube has side length
  $\ell \leq 1$ can be packed  into a unit bin if
  ${\rm Vol}(S) \leq \ell^d +(1-\ell)^d$, 
  where ${\rm Vol}(S)$ is the total volume of the hypercubes in $S$.
\end{theorem}

\begin{proposition}
  \label{thm:PoA_upbd}
  For the prices of anarchy of the selfish $d$-hypercube bin
  packing game, we have 
  \begin{equation}
    \label{eq:PoA_upbd}
    \SPoA(d)\leq\PoA(d)\leq{2^d}.
  \end{equation}
\end{proposition}
\begin{proof}
  The fact that~$\SPoA(d)\leq\PoA(d)$ follows directly from the
  definitions.  We therefore only address the second inequality
  in~\eqref{eq:PoA_upbd}. 
  Let $\cP=\{B_1,\ldots,B_k\}$ be a packing that is a Nash
  equilibrium, where each $B_i$ is a packing into a unit bin. The
  proof is based on volume arguments. We prove that each bin in $\cP$
  has volume occupation of at least $1/2^d$, except possibly for one bin.
  For simplicity, if $B$ is a packing, we denote by ${\rm Vol}(B)$ the total
  volume of the hypercubes in $B$, and by $H(B)$ the set of hypercubes
  in $B$.  Let $\cP'=\{B\in\cP:\,{\rm Vol}(B)<1/2^d\}$. We will prove
  that $|\cP'|$ is at most~$1$.  First, note that, if $B\in\cP'$ then,
  the side length of any hypercube in $B$ is smaller than $1/2$. The
  proof is by contradiction.  Suppose there are at least two distinct
  bins $B'$, $B''\in\cP'$ such that
  $ 1/2^d>{\rm Vol}(B')\geq {\rm Vol}(B'')$, and let $b\in
  H(B'')$.
  Since $\cP$ is a Nash equilibrium, $H(B')\cup \{b\}$ cannot be
  packed in only one bin. If
  $\ell=\max\{s(h):\,h\in H(B')\cup\{b\}\}$, by Theorem~\ref{NewRef1}
  we conclude that
  \begin{equation}
    {\rm Vol}(H(B')\cup\{b\}) > \ell^d+(1-\ell)^d.
  \end{equation}
  Thus, 
  \begin{eqnarray}
    {\rm Vol}(B') > (1-\ell)^d + \ell^d - {\rm Vol}(\{b\}) \geq (1-\ell)^d \geq 1/2^d,
  \end{eqnarray}
  a contradiction to the choice of $B'$. Therefore, there is at most
  one bin in $\cP'$ (and consequently at most one bin in $\cP$ with
  volume occupation smaller than $1/2^d$) and we can conclude that $|\cP| \leq  2^d
  \,{\rm Vol}(H(\cP))  + 1\leq  2^d\,\opt(H(\cP)) + 1$. This completes the proof.
\end{proof}

Exponentially better upper bounds than the one in
Proposition~\ref{thm:PoA_upbd} can be proved for~$\SPoA(d)$.  However,
we do not quite see how to close the exponential gap between our lower
and upper bounds (in fact the gap is doubly exponential in the case
of~$\SPoA(d)$).  We shall address these topics elsewhere.

\end{document}